\documentclass[
  aps,
  pra,
  reprint,
  superscriptaddress,
  nofootinbib,
  longbibliography
]{revtex4-2}

\usepackage{xcolor}
\definecolor{blueviolet}{rgb}{0.2, 0.2, 0.6}
\definecolor{webgreen}{rgb}{0,.5,0}
\definecolor{webbrown}{rgb}{.6,0,0}
\usepackage[pdftex,
	bookmarks=false,
	colorlinks=true, 
	urlcolor=webbrown, 
	linkcolor=blueviolet, 
	citecolor=webgreen,
	pdfstartpage=1,
	pdfstartview={FitH},  
	bookmarksopen=false
	]{hyperref}

\usepackage{amsmath,amssymb,amsfonts}
\usepackage{mathtools}
\usepackage{bm}
\usepackage{braket}
\usepackage{amsthm}

\usepackage{graphicx}
\usepackage{booktabs}
\usepackage{array}
\usepackage{subcaption}
\usepackage{tikz}
\usetikzlibrary{arrows.meta,positioning,calc,fit}

\usepackage[T1]{fontenc}
\usepackage{microtype}

\usepackage{algorithm}
\usepackage{algpseudocode}

\newcommand{\CZ}{\operatorname{CZ}}
\newcommand{\CNOT}{\operatorname{CNOT}}
\newcommand{\cG}{\mathcal G}

\newtheorem{theorem}{Theorem}
\newtheorem{lemma}[theorem]{Lemma}
\newtheorem{proposition}[theorem]{Proposition}
\newtheorem{corollary}[theorem]{Corollary}
\theoremstyle{definition}
\newtheorem{definition}[theorem]{Definition}

\begin{document}

\title{(A Variant of) Clifford Circuit Synthesis is NP-Complete}

\author{Luna Lima Keller}
\author{Richard Kueng}
\affiliation{Department for Quantum Information and Computation at Kepler (QUICK), Johannes Kepler University Linz, Austria
}

\date{\today}

\begin{abstract}

Optimal circuit synthesis is the problem of finding the shortest-depth circuit representation of a given functionality with respect to a pre-specified elementary gate set. This is a central problem in both quantum and classical hardware design. While the classical version is very well understood – both in terms of heuristics and rigorous hardness assertions – much less is known about optimal quantum circuit synthesis.

We focus on optimal Clifford circuit synthesis, for which various heuristics are known, e.g. via reduction to 3-SAT. Our main result supplies a matching hardness result: a variant of optimal Clifford synthesis is NP-hard. The proof proceeds in two parts: (i) reduce 3-edge colorability on 3-regular graphs to a circuit synthesis problem that only involves CZ gates, (ii) prove that the availability of additional elementary Clifford gates -- most notably: Hadamard, phase and CNOT -- cannot lead to further improvements of the optimal circuit depth. 
Our work sharpens the complexity-theoretic understanding of Clifford circuits: simulation and equivalence checking are in P, whereas deciding whether a Clifford unitary admits an implementation within a prescribed depth is NP-complete.
\end{abstract}

\maketitle


\begin{figure*}[t]
    \centering
    \includegraphics[
        width=\textwidth
    ]{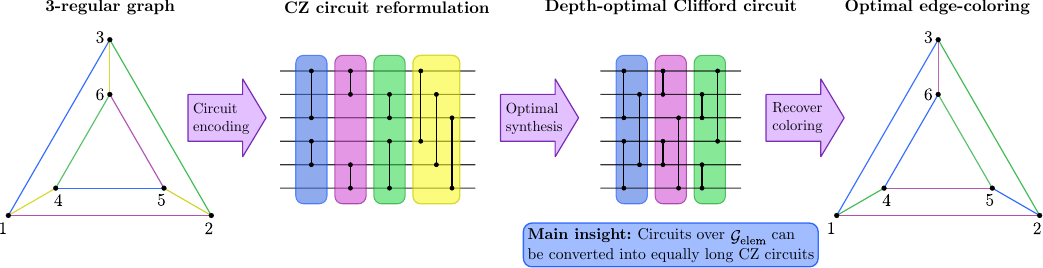}
    \caption{\emph{Caricature of the Karp reduction from 3-regular graph edge coloring to depth-optimal Clifford circuit synthesis.} A 4-coloring of the initial graph with $n=6$ vertices is included for illustration.
    }
    \label{fig:illustration}
\end{figure*}

\section{Introduction}

Decomposing a given logical functionality into a sequence of elementary gates is known as \emph{optimal circuit synthesis}. A solution is (depth-)optimal if it is impossible to achieve the desired
functionality with even shorter circuits. Optimal circuit synthesis is a
central problem in classical hardware design and the quantum version arises naturally when mapping a high-level quantum circuit functionality to elementary gate instructions.

 It should also come as no surprise that optimal circuit synthesis is a hard problem.
Already the simpler task of deciding whether two circuits involving $n$ (qu-)bits are equivalent is a daunting task, because their outputs may have to be compared for exponentially many inputs (in $n$). In the case of (classical) Boolean circuits,  equivalence checking is known to be co-NP complete~\cite{arora2009}. Optimal synthesis is then even harder, because you optimize depth while relying on equivalence checking as a vital subroutine. Formally, optimal Boolean circuit synthesis is contained in the second level of the polynomial hierarchy. And it is also a complete problem for this complexity class~\cite{UMANS2001,BUCHFUHRER2011}. 

For quantum circuits, things become even more involved. 
Approximate equivalence checking is known to be QMA-complete~\cite{Janzing2005,Ji2009}, while exact identity testing is complete for the class Nondeterministic Quantum Polynomial-time (NQP)~\cite{Tanaka2010}. Optimal circuit synthesis then not only asks for approximately reproducing a given functionality, but for minimizing circuit depth. 
While we don't have as clear a grasp on the hardness of this problem as in the classical case, certain hardness assertions are known: In some settings, optimizing different metrics of Clifford+$T$ circuits is NQP-hard, while the corresponding decision problems lie in $\mathrm{NP}^{\mathrm{NQP}}$~\cite{vandewetering2024,kjelstrøm2026}. 
In more restricted settings, sharper results have been obtained: various CNOT synthesis problems are known to be NP-complete,~\cite{Amy2019,Jiang2020,li2026,acuaviva2026} and a large body of work has developed exact and heuristic algorithms for CNOT circuit synthesis~\cite{debrugiere2021,murphy2023,markov2008,prasad2006,schaeffer2014,Webster2026}.

Here we focus on \emph{Clifford circuits} which, 
play a pivotal role in quantum error correction, fault-tolerant quantum computing, measurement-based quantum computation, randomized benchmarking and more~\cite{Lidar2013,raussendorf2001,magesan2012,huang2020,Linden2021,Burgholzer2021}.
Clifford circuit synthesis has also received substantial attention~\cite{Aaronson2004,bravyi2022,Peham2023,shaik2025,Webster2026,amy2022,winderl2023}.
What makes them interesting in the context of computational complexity is the Gottesman-Knill theorem: the action of Clifford circuits on stabilizer states can be efficiently simulated on a classical computer~\cite{gottesman1998,gidney2021,brandl2026}. This, in turn, gives rise to efficient classical algorithms that check whether two Clifford circuits are equivalent. 
Ref.~\cite{Peham2023} used such an efficient Clifford equivalence checking subroutine to reduce the more involved task of depth-optimal Clifford synthesis to an instance of the satisfiability problem (SAT). So, while the classical optimal circuit synthesis problem is complete for the second level of the polynomial hierarchy, the optimal Clifford circuit synthesis problem is guaranteed to be contained one polynomial hierarchy level lower in NP. This is not only of conceptual interest, but also has practical implications: heuristic SAT solvers perform very well in practice and can thus be employed to find Clifford circuit realizations of optimal depth~\cite{shaik2025}. In a similar vein, Ref~\cite{Lee2026} uses empirical graph coloring solvers to achieve depth optimal synthesis for commuting quantum circuits.

But it could be that such reformulations as SAT or graph coloring artificially inflate the intrinsic difficulty of optimal Clifford circuit synthesis. After all, the latter are known to be NP-hard, while no definitive hardness assertion for optimal Clifford circuit synthesis is currently known. Here, we fill this gap: Theorem~\ref{thm:main} below states that optimal Clifford circuit synthesis is itself NP-hard.

From a practical perspective, our findings lend credence to existing Clifford circuit synthesis heuristics. The intrinsic complexity remains unchanged if we recast this as instances of other NP-hard problems, like SAT~\cite{Peham2023,shaik2025} or graph coloring~\cite{Lee2026}, for which powerful empirical solvers exist~\cite{biere2024,Timo2026}.
From a conceptual perspective, our result completely resolves the computational complexity of an important subclass of quantum optimal circuit synthesis.

\section{Summary of main results}

The main result of this work is a hardness assertion for finding the shortest-depth realization of a given $n$-qubit Clifford circuit. To make definitions precise, we assume that gates on disjoint qubits may be executed in parallel, in a single layer. For a fixed elementary gate set $\cG$, let
\begin{equation*}
 D_{\cG}(U)=\min\{d: U=L_d\cdots L_1\},
 \label{eq:depth}
\end{equation*}
where each layer $L_t$ subsumes a tensor product of gates from $\cG$ with pairwise disjoint supports.

\begin{definition}[Optimal Clifford circuit synthesis] \label{def:synthesis_problem}
     Given a $n$-qubit target Clifford unitary $U$, determine the minimum number of parallel gate layers required to implement $U$ over a fixed elementary gate set $\cG$.
\end{definition}

Note that Definition~\ref{def:synthesis_problem} depends on the type of elementary gate set used. Throughout this work, we will consider the following \emph{elementary Clifford gate set}:
\begin{align}
    \cG_{\mathrm{elem}}
    &=
    \{X,Y,Z,H,S,S^\dagger,\CNOT,\CZ\}.
    \label{eq:gate-set}
\end{align}
Note that this set is larger than the minimal gate set comprised of $H,S$ and CNOT (or CZ). 
Our main result states the following.

\begin{theorem} \label{thm:main}
Optimal Clifford circuit synthesis (Definition~\ref{def:synthesis_problem}) is NP-hard for the elementary gate set $\cG_{\mathrm{elem}}$.
\end{theorem}

This means that every other problem in the complexity class NP can be reduced to an instance of optimal Clifford circuit synthesis with at most polynomial overhead. This nicely complements prior work that showed that (the decision version\footnote{
Can one synthesize a given $n$-qubit Clifford circuit with depth $d$?} of) optimal Clifford circuit synthesis is itself contained in NP~\cite{Peham2023}. The following corollary is now an immediate consequence and fully resolves the computational complexity of this restricted class of optimal quantum circuit synthesis problems. 

\begin{corollary}
A variant of optimal Clifford circuit synthesis is NP-complete.
\end{corollary}

The proof of Theorem~\ref{thm:main} proceeds in two steps.

First, we relate \emph{graph edge coloring} to a very restricted instance of optimal Clifford circuit synthesis that only involves $\CZ$ gates.
A proper edge coloring of a graph $G$ is a color assignment on its edges such that any pair of incident edges have different colors. The minimum number of colors is known as the edge chromatic number $\chi'(G)$. Computing this value is NP-hard, even for the more restricted class of degree-three graphs~\cite{Holyer1981}.
The main ingredient relating edge coloring to a Clifford circuit is as follows: First, associate every vertex to a qubit, then for every edge in $G$ we include a $\CZ$ gate between the corresponding qubits. We denote this unitary by 
\begin{equation}
    U_G
    =
    \prod_{uv\in E(G)}\CZ_{uv},
    \label{eq:UG-intro}
\end{equation}
and refer to Figure~\ref{fig:illustration} for a visual illustration. 
Our first technical contribution relates the edge chromatic number of $G$ to the optimal circuit depth of $U_G$.

\begin{proposition}[Edge chromatic number equals optimal depth of CZ-circuits] 
\label{prop:edge-chromatic-number}
The edge chromatic number of any graph $G$ can be computed by determining the optimal depth of a CZ-circuit implementing $U_G$: $\chi' (G)=D_{\CZ}(U_G)$.
\end{proposition}

Since computing edge-chromatic numbers is NP hard, this in turn establishes NP-hardness for a restricted case of optimal Clifford circuit synthesis, where the elementary gate set is restricted to $G_{\CZ}= \left\{\CZ\right\}$. 
The next step is to extend (a variant of) this hardness assertion to the full elementary gate set of the Clifford group. The following statement shows that this is possible if we restrict attention to 3-regular graphs.

\begin{proposition}[No shortcuts with more Clifford gates] 
\label{prop:no-shortcuts}
Let $G$ be a 3-regular graph. Then, $D_{\mathrm{elem}}(U_G)=D_{\CZ}(U_G)$. 
\end{proposition}

Importantly, this restriction to 3-regular graphs does not change NP-hardness of the original edge-coloring task~\cite{Holyer1981}. Theorem~\ref{thm:main} now readily follows from combining both Propositions to conclude $D_{\mathrm{elem}}(U_G)=D_{\CZ}(U_G) = \chi'(G)$ for every 3-regular graph.

\section{Methods}
\label{sec:methods}

This section provides proof sketches for the two main technical contributions of this work. We deliberately put an emphasis on intuition and guidance. Fully fleshed versions of these proofs can be found in the appendix.


\subsection{Edge chromatic number equals optimal depth of $\CZ$-circuits (Proposition~\ref{prop:edge-chromatic-number}) }

Proposition~\ref{prop:edge-chromatic-number} reports an equivalence between the edge chromatic number $\chi' (G)$ of a graph $G$ and the optimal circuit depth of the related unitary $U_G = \prod_{uv\in E(G)}\CZ_{uv}$.


The direction $D_{\CZ}(G) \leq \chi' (G)$ readily follows from assigning each CZ gate to a layer corresponding to the edge coloring: Given an edge coloring $uv\mapsto c(uv)\in\{1,2,3,4\}$, put each $\CZ_{uv}$ in layer $L_{c(uv)}$:
\begin{equation*}
L_t=\prod_{c(uv)=t}\CZ_{uv}.
\end{equation*}
Since edges in the same color class are not incident, no two CZ gates share a qubit and they can be executed in parallel. The converse direction $D_{\CZ}(G) \leq \chi' (G)$ is less straightforward, but follows from the following auxiliary result.

\begin{lemma}
\label{lem:methods-cz-parity}
In any CZ-only realization of $U_G$, the gate $\mathrm{CZ}_{uv}$ occurs an odd number of times if $uv\in E(G)$, and an even number of times otherwise.
\end{lemma}

\begin{proof}[Proof sketch]
CZ gates commute and cancel in pairs, so their combined action depends only on these parities. Each parity can also be recovered from the unitary: on a $n$-qubit computational-basis input whose only nonzero bits are at $u$ and $v$, only $\mathrm{CZ}_{uv}$ contributes a minus sign.
\end{proof}

Now note that CZ gates always commute with each other and applying the same CZ gate twice cancels its action. This allows us to turn the content of Lemma~\ref{lem:methods-cz-parity} into a potentially shorter CZ-circuit. Indeed, whenever $\mathrm{CZ}_{uv}$ occurs an even number of times, we can commute all these gates through the circuit and have them cancel. If $\mathrm{CZ}_{uv}$ instead occurs an odd number of times, the same strategy leads to exactly one of these gates remaining. It is important to note that this strategy can never increase circuit depth and produces a CZ-only unitary that also implements $U_G$.
Since gates within the same layer cannot overlap, the CZ gates in each layer now correspond to non incident edges in the original graph. So, converting the layer label of each CZ gate into a color of the associated edge must result in a proper edge coloring. This establishes $D_{\CZ}(G) \leq \chi' (G)$. Appendix~\ref{app:cz} contains a formal proof of Lemma~\ref{lem:methods-cz-parity}.

\subsection{No shortcuts with more Clifford gates (Proposition~\ref{prop:no-shortcuts})}

This section establishes the ideas behind the main technical statement of this work: allowing additional Clifford gates cannot reduce the optimal depth of $\CZ$-circuits that arise from graph edge colorings, provided that the underlying graph has degree three. 
In formulas, 
\begin{equation*}
 D_{\rm elem}(U_G)=D_{\rm CZ}(U_G).
 \label{eq:elem-equality}
\end{equation*}
for every 3-regular graph $G$.


First, we note that, by Vizing's theorem~\cite{Vizing1964}, the edge chromatic number of 3-regular graphs $G$ can only assume two values: $\chi'(G) \in \left\{3,4\right\}$. This, together with the trivial inequality $D_{\rm elem}(U_G)\leq D_{\rm CZ}(U_G)$, enables us to consider two cases: 

If a given graph $G$ is such that $D_{\rm elem}(U_G)=4$, we immediately have $\chi'(G) = 4$. In other words, our reduction correctly maps graphs which are not $3$-colorable to unitaries that cannot be implemented in $3$ layers.

We then can focus on proving that every graph with $D_{\rm elem}(U_G)\leq3$ corresponds to 3-colorable graphs, namely, we show that $D_{\rm elem}(U_G)\leq3\Rightarrow D_{\rm elem}(U_G)=D_{\rm CZ}(U_G)$.

\begin{lemma}[Hadamard elimination]
\label{lem:had-elim}
Every realization of $U_G$ of depth at most three over $\mathcal G_{\mathrm{elem}}$ can be made Hadamard-free without increasing depth.
\end{lemma}

We describe the argument using $n$-qubit stabilizer bases. These are  joint eigenvectors of maximal commuting families of $n$-qubit Pauli operators. 
A concrete example is the computational basis which is the joint eigenbasis of all $2^n$ Paulis generated by $Z_1=Z \otimes I^{\otimes (n-1)}, Z_2=I \otimes Z \otimes I^{\otimes (n-2)}, \ldots, Z_n=I^{\otimes (n-1)} \otimes Z$. Other bases of this type arise if we replace certain Pauli-$Z$s with $X$ instead. Every such \emph{$X/Z$ product basis} can be produced by starting with the computational basis and applying a single layer of Hadamard and identity gates. Note that there are $2^n$ such bases in total.


We say a basis is preserved by the action of a given Clifford gate, if its vectors are at most permuted and/or multiplied by phases. With the exception of Hadamard gates, all other elementary gates in $\cG_{\rm elem}$ preserve the computational basis. Indeed, CNOT and Pauli gates permute the basis vectors, while the diagonal gates $S,S^\dagger$ and $\CZ$ only add phases.

The first ingredient in our argument is the observation that a single Hadamard-free layer cannot change one $X/Z$ product  basis into a different one.

\begin{lemma}
\label{lem:methods-basis}
If a single Hadamard-free layer maps a $X/Z$ product basis $\mathcal{B}$ to another $X/Z$ product basis $\mathcal{B'}$, then $\mathcal{B}=\mathcal{B'}$.
\end{lemma}

\begin{proof}[Proof sketch]
This follows from tracking the possible action of the action of elementary Clifford gate. For example, $\CNOT_{12}$ preserves the two-qubit product bases defined by $\{Z_1,Z_2\}$, $\{X_1,X_2\}$ and $\{Z_1,X_2\}$ as a whole, but it can permute the vectors within each basis. However, its action on the $\{X_1,Z_2\}$ product basis produces an entangled stabilizer basis, which is not a $X/Z$ product basis anymore. Similarly, an $S$ gate preserves the $Z$ eigenvector but exchanges $X$ and $Y$ eigenvectors. None of these operations changes one $X/Z$ product basis into a different $X/Z$ product basis. Similar arguments hold for $\CZ$ and all Pauli operators.
\end{proof}

With this auxiliary result at hand, let us return to the task at hand. We take the envisioned depth-$3$ decomposition of $U_G$ and decompose it further as
\begin{equation*}
 U_G=(M_3H_{S_3})(M_2H_{S_2})(M_1H_{S_1})
\end{equation*}
where for $t=1,2,3$,  $H_{S_t}$ subsumes all Hadamards in layer $t$ and $M_t$ contains all other gates in the same layer. 
Since constituents of the same layer commute  ($[M_t,H_{S_t}]=0$), we can rewrite this as $U_G= M_3 H_{S_3} M_2 H_{S_2} H_{S_1} M_1$.

We now track how the computational basis changes layer by layer. The discussion before Lemma~\ref{lem:methods-basis} ensures that $M_1$ (which is Hadamard-free) preserves this basis, while the central Hadamard layer $H_{S_2} H_{S_1}$ exchanges the local product bases ($X \leftrightarrow Z$) for each qubit on which a Hadamard gate acts nontrivially. The result is therefore another $X/Z$ product basis. Since $U_G$ preserves the computational basis by construction, running the argument backward from the output through the last layer gives another $X/Z$ product basis.
Finally, the remaining action of the $M_2$ layer (which is Hadamard-free) is handled by Lemma~\ref{lem:methods-basis} and must preserve the basis.

Thus, since all layers of the type $M_t$ leaves the bases on which they act invariant, the combined action of the Hadamard layers must also leave the computational basis invariant. In other words, each qubit is acted on by either $0$ or $2$ Hadamard gates.
Each Hadamard therefore either cancels with a copy in an adjacent layer or pairs with a copy in the opposite outer layer to conjugate a single middle-layer gate.
By following which gates are allowed to appear in $M_2$, we can then check exhaustively that each conjugation results in another gate in $\cG_{\rm elem}$. These replacements do not change gate supports. Hence,  they preserve the layer structure and also do not increase depth, which establishes Lemma~\ref{lem:had-elim}. Appendix~\ref{app:hadamard-free} formalizes this part of the argument.

Now that we have eliminated Hadamard gates from an optimal Clifford circuit synthesis of $U_G$, we move on to do the same for CNOT gates. To achieve this goal, we first convert this circuit into a convenient normal form. 

\begin{lemma}[CNOT normal form]
\label{lem:methods-cnot-normal}
A Hadamard-free realization of $U_G$ of depth at most three can be simplified, without increasing depth, so that its middle layer contains no CNOTs and its outer layers contain identical CNOTs with the same orientation.
\end{lemma}

\begin{proof}[Proof sketch]

A Hadamard-free circuit acts on computational-basis inputs by changing their bit values and accumulating phases. CNOTs mix input bits by XOR, which accounts for a linear transformation, while $X$ and $Y$ may also flip individual bits. Since $U_G$ returns every input string unchanged, the overall mixing of the CNOTs must be trivial.

The proof follows from the following fact: if a qubit participates in two unequivalent CNOTs in the first two layers, then, to undo the resulting linear action on the full bitstring, it would need to participate in at least two more overlapping CNOTs, which is impossible in three layers. Thus every CNOT in the middle layer is removable by composing with a identical CNOT in an adjacent layer.

Then, since the last layer must exactly undo the linear action of the first layer, they are forced to have identical CNOTs, which is the normal form stated in the lemma.
\end{proof}

The next statement uses structural properties of $U_G$ for 3-regular graphs $G$ to remove CNOTs entirely.

\begin{lemma}[CNOT elimination for 3-regular graphs]
\label{lem:cnot-elim}
If $G$ is 3-regular, the normal form of
Lemma~\ref{lem:methods-cnot-normal} contains
no CNOTs.
\end{lemma}

\begin{proof}[Proof sketch]

We prove by contradiction: Suppose that $\CNOT_{a\to b}$ is present in both outer layers and consider a variant of the original circuit circuit $W=\CNOT_{a\to b}U_G=\CNOT_{a\to b} L_3 L_2 L_1$, where we cancel out the CNOT between qubit $a$ and qubit $b$ on the very left, but the same CNOT on the very right remains.  By taking advantage of the fact that $C$ fixes $X_b$ by conjugation, we arrive at two different predictions for the Pauli weight of $WX_b W^\dagger$.
The original $CZ$ circuit propagates $X_b$ to all its neighbors: $\operatorname{wt}(WX_b W^\dagger)=1+\deg_G(b)=4$ given that $G$ is 3-regular. On the other hand, by taking advantage of the simple structure of the RHS, the computed weight is \textit{at most} $3$ -- a contradiction.

\end{proof}

After having removed both Hadamard gates (Lemma~\ref{lem:had-elim}) and CNOT gates (Lemma~\ref{lem:cnot-elim}), are now in the final stages of establishing Proposition~\ref{prop:no-shortcuts}: any depth-three implementation of $U_G$ can be converted to a circuit $C$ containing only $\CZ$-gates and single-qubit Clifford gates different from Hadamard. The definition of $U_G$ in Eq.~\eqref{eq:UG-intro} ensures the following action on a $n$-qubit computational basis state $\ket{x_1,\ldots,x_n}$ :
\begin{equation*}
    U_G\ket{x_1,\ldots,x_n}
    =(-1)^{\sum_{\{u,v\}\in E(G)}x_u x_v}\ket{x_1,\ldots,x_n}.
\end{equation*}
Note that each CZ gate contributes a single term to the quadratic phase. Also, single-qubit gates cannot contribute quadratically to this overall phase. hence, each term in the quadratic sum must correspond to an actual CZ gate appearing in the circuit $C$. Since the underlying graph is 3-regular, it must have $3n/2$ edges in total. Also, a single circuit layer may contain at most $n/2$ (non-overlapping) CZ gates and we know that there is exactly 3 such layers. Hence, a simple counting argument forces the final depth-three circuit to populate all layers with exclusively CZ-gates. This establishes the claim in Proposition~\ref{prop:no-shortcuts}: $D_{\rm elem}(U_G)=D_{\rm CZ}(U_G)$.

\section{Synopsis}

In this work we have shown that optimal $n$-qubit Clifford circuit synthesis is, in general, at least as challenging as any other problem in NP, see Theorem~\ref{thm:main}. This complements existing work in the literature: \cite{Peham2023} shows that every optimal Clifford circuit synthesis problem can be mapped to a satisfiability problem with polynomial overhead. Hence, the optimal Clifford circuit synthesis problem is itself contained in NP.

We establish NP-hardness of optimal Clifford circuit synthesis via reduction from graph edge coloring. The proof proceeds in two steps and heavily relies on structural properties of (very) short Clifford circuits. We first converted edge coloring of a graph with $n$ vertices to an instance of optimal $n$-qubit circuit synthesis with only CZ gates. This was the content of Proposition~\ref{prop:edge-chromatic-number}. Then, we restricted our attention to 3-regular graphs, where the edge chromatic number is known to be either $3$ or $4$. Subsequently, we showed that additional Clifford gates -- most notably Hadamards, CNOT and phase gates -- cannot provide any synthesis shortcuts over CZ circuits. At least if the circuit in question is already very short and there is an underlying graph structure present, see Proposition~\ref{prop:no-shortcuts}.

This proof technique, at least in its current form, depends on the type of elementary gate set used for the Clifford group. As displayed in Eq.~\eqref{eq:gate-set}, we include CZ gates as a primitive depth-one gate. This differs slightly from implementations using only $H,S,S^\dagger,\CNOT$ and Pauli gates as primitives. SAT-based optimal Clifford circuit synthesis frameworks can be extended to gate sets including CZ~\cite{Peham2023}, but translating the present reduction to a library without primitive CZ requires some modifications. A natural extension of the current result is to count an arbitrary one- or two-qubit Clifford as a primitive gate. The present proof uses the finer elementary structure of $H,S,\CNOT$, and CZ and therefore does not immediately extend to that model. We leave this to future work.

A very similar setting has recently emerged for \emph{matchgate circuits}, another non-universal subclass of quantum circuits that admits efficient classical simulation~\cite{jozsa2008}. A $n$-qubit matchgate unitary can be compactly represented by a $2n\times 2n$ orthogonal matrix. This structure has recently been leveraged for synthesis tasks, e.g., Ref.~\cite{Braccia2025} constructs circuits for sampling Haar-random fermionic linear-optics unitaries with optimal depth and gate count, while Ref.~\cite{casas2026,casas2026optimal} addresses optimal exact matchgate synthesis, reducing this problem to SAT, in a similar vein to Ref.~\cite{Peham2023}. These results mirror the Clifford setting closely, suggesting that the structural ideas developed here may also be adaptable to the matchgate setting. This could result in a rigorous proof of NP-hardness for optimal matchgate synthesis.

\subsection*{Acknowledgments}

We thank Diego García-Martín for inspiring discussions and feedback.
This work has been supported by  the Austrian Science Fund (FWF) via the START award q-shadows (713361001) and the SFB BeyondC
(10.55776/FG7) (research group 7), as well as the European Research Council via the ERC Starting grant q-shadows (101117138).

\bibliography{ref}

\cleardoublepage
\appendix

\section{Optimal synthesis of CZ-only circuits}
\label{app:cz}

Here we prove that we can always extract an edge-coloring from a CZ-circuit implementing $U_G$. The first step is Lemma~\ref{lem:methods-cz-parity}, which we restate for convenience.

\begin{lemma}
In any CZ-only realization of $U_G$, the gate $\mathrm{CZ}_{uv}$ occurs an odd number of times if $uv\in E(G)$, and an even number of times otherwise.
\end{lemma}

\begin{proof}
    
Consider any CZ circuit of depth $d$ implementing $U_G$ as defined in Eq.~\eqref{eq:UG-intro}. Since CZ gates commute and square to the identity, its action is determined solely by the parity $m_{uv}$ of the number of occurrences of each gate $\CZ_{uv}$. There is a one-to-one correspondence between unitaries $U_G$ and mappings $uv\mapsto m_{uv}\in \{0,1\}$. One way to see this is to evaluate the circuit on computational-basis states $\ket{x}=\ket{x_1,\ldots,x_n}$:
\begin{align*}
    U_G \ket{x}=&(-1)^{\sum_{uv\in E(G)}x_u x_v}\ket{x} \quad \text{but also} \\
    U_G\ket{x}= &(-1)^{\sum_{u<v} m_{uv}x_u x_v}\ket{x}.
\end{align*}
Here, the first equation corresponds to the original gates in Eq.~\eqref{eq:UG-intro} and the second equation describes the action of the depth-$d$ circuit on the computational basis.

Each odd occurrence of $\CZ_{uv}$ contributes with a quadratic term $x_ux_v$. The uniqueness of these coefficients forces odd multiplicity on graph edges and even multiplicity when no edge connects two given vertices.

\end{proof}

If we retain one occurrence of each graph edge and delete every other gate, then each remaining layer is still valid and the layer labels form a proper edge coloring using at most $d$ colors. Hence $\chi'(G)\le d$.

This establishes the exact correspondence between CZ circuits implementing $U_G$ and edge colorings, without assuming that an optimal circuit merely reorders the supplied gates, which is the content of Proposition~\ref{prop:edge-chromatic-number}.

\section{Hadamard-free optimal circuits}
\label{app:hadamard-free}

We now begin the sequence of statements needed for proving that optimal Clifford circuits over $\cG_{\rm elem}$ are not shorter than CZ-circuits for implementing $U_G$.

Let us start with notation. We regard a basis as an unordered collection of one-dimensional subspaces: phases and permutations of its vectors do not change the basis. Let $\mathcal B_Z$ denote the computational basis. For $S\subseteq[n]$, define
\begin{equation*}
    \mathcal B_S=H_S\mathcal B_Z,
    \qquad
    \text{where }H_S=\prod_{i\in S}H_i.
\end{equation*}
Thus $\mathcal B_S$ is the tensor product basis that uses the $X$ eigenbasis on qubits in $S$ and the $Z$ eigenbasis on the remaining qubits.

\begin{lemma}
\label{lem:coordinate}
Let $M$ be a single layer of Hadamard-free elementary
gates. If
\begin{equation*}
    M\mathcal B_R=\mathcal B_T,
\end{equation*}
then $R=T$.
\end{lemma}

\begin{proof}
Since the gates in $M$ have disjoint supports and each basis consists of fully separable states, we can analyze the action of each possible gate in $M$ on its corresponding subspace.

Let us analyze single-qubit gates first: Pauli gates only permute or add phases to $X$ and $Z$ eigenbasis, preserving both $\{\ket{0},\ket{1}\}$ and $\{\ket{+},\ket{-}\}$. The action of $S$ and $S^\dagger$ preserves the $Z$ basis, but maps $X$ eigenvectors to $Y$ eigenvectors. By the hypothesis of the lemma, $M$ is forced to not having $S$ or $S^\dagger$ acting on any qubit in $R$.

Let us now consider two-qubit gates. Consider $a,b\in\{0,1\}$ and denote $\ket{(-)^{0}}:=\ket{+}$, $\ket{(-)^{1}}:=\ket{-}$. The $\CNOT_{12}$ gate preserves three out of the four possible X/Z product bases:
\begin{align*}
    &\CNOT_{12}|a\rangle|b\rangle
        =|a\rangle|a\oplus b\rangle,\\
    &\CNOT_{12}|a\rangle|(-)^b\rangle
        =(-1)^{a+b}|a\rangle|(-)^b\rangle\rangle,\\
    &\CNOT_{12}|(-)^a\rangle|(-)^b\rangle
        =|(-)^{a+b}\rangle|(-)^b\rangle.
\end{align*}
The remaining basis is mapped to a Bell-type basis of entangled two-qubit states which cannot be equivalent to a X/Z product basis:
\begin{equation*}
    \CNOT_{12}|(-)^a\rangle\otimes|b\rangle
    =
    \frac{|0\rangle|b\rangle
          +(-1)^a|1\rangle|b\oplus1\rangle}
         {\sqrt{2}}.
\end{equation*}
Using the notation established above, these bases are respectively $\mathcal{B}_Z,\mathcal{B}_{\{2\}},\mathcal{B}_{\{1,2\}}$ and $\mathcal{B}_{\{1\}}$. We summarize the above discussion in table \ref{tab:CNOT-CZ}, where we also include analogous conclusions for the action of a  $\CZ$ fate. 

\begin{table}[h]

\caption{Action of $\CNOT_{12}$ and $\CZ$ on X/Z product bases.}
\label{tab:CNOT-CZ}
\begin{tabular}{c|cccc}
 & $\mathcal{B}_Z$ & $\mathcal{B}_{\{1\}}$ & $\mathcal{B}_{\{2\}}$ &  $\mathcal{B}_{\{1,2\}}$ \\
\hline
$\mathrm{CNOT}_{12}$
 & preserved & \textit{entangled} & preserved & preserved \\
$\mathrm{CZ}$
 & preserved & preserved & preserved & \textit{entangled}
\end{tabular}
\end{table}

We now have all ingredients to handle the full $n$-qubit case. Since $M$ maps a X/Z basis to another X/Z basis, then it must not contain a $\CNOT_{a\to b}$ where $a\in R$ and $b\in R^c$, and similarly, it must not contain a $\CZ$ between two qubits in $R$. Since every other option \textit{preserves} each basis, the conclusion follows.

\end{proof}

We now prove that Hadamards can be eliminated completely from the optimal circuit.

\begin{lemma}[Hadamard elimination]
\label{lem:app-had-elim}
Every realization of $U_G$ of depth at most three over $\mathcal G_{\mathrm{elem}}$ can be made Hadamard-free without increasing depth.
\end{lemma}

\begin{proof}

Suppose
\begin{equation*}
    U_G=L_3L_2L_1
\end{equation*}
is a depth-three realization, where each $L_t$ is a layer of pairwise disjoint gates drawn from $\mathcal G_{\rm elem}$. For each $t\in\{1,2,3\}$, let $R_t\subseteq[n]$ be the set of qubits in which $L_t$ contains a Hadamard gate. Since the gates within one layer have pairwise disjoint supports, no other gate in $L_t$ acts on a qubit in $R_t$. We may therefore write
\begin{equation*}
    L_t=M_tH_{R_t},
    \qquad \text{where} \qquad
    H_{R_t}:=\prod_{i\in R_t}H_i,
    \label{eq:layer_hadamard_decomposition}
\end{equation*}
and $M_t$ is a Hadamard-free layer disjoint from $H_{R_t}$.

Note that every Hadamard-free elementary gate preserves the computational basis as a set. Therefore
\begin{align*}
    L_2L_1\mathcal B_Z&=(M_2H_{R_2})(M_1H_{R_1})\mathcal B_Z\\
    &=M_2(H_{R_2}H_{S_1})M_1\mathcal B_Z\\
    &=M_2H_{R_1\triangle R_2}\mathcal B_Z\\
    &=M_2\mathcal B_{R_1\triangle R_2},
\end{align*}
where $\triangle$ denotes the symmetric difference of sets $A\triangle B=(A\setminus B)\cup (B\setminus A)$.

Likewise, since $U_G$ preserves $\mathcal B_Z$,
propagating backward through the last layer gives
\begin{equation*}
    L_2L_1\mathcal B_Z
    =L_3^{-1}\mathcal B_Z
    =\mathcal B_{R_3}.
\end{equation*}
Combining these observations with Lemma~\ref{lem:coordinate} implies
\begin{equation*}
    R_1\triangle R_2 =R_3.
\end{equation*}
This ensures that every qubit carries an even number of Hadamards across the three layers.

By collapsing Hadamards on neighboring layers,
we can assume that the middle layer is Hadamard free, and we have
\begin{equation*}
    R_1=R_3=:R,
    \qquad
    R_2=\varnothing.
\end{equation*}
The circuit therefore takes the form
\begin{equation*}
    U_G=
    (M_3 H_R) M_2 (H_R M_1)=
    M_3
    \left(
        H_R M_2 H_R
    \right)
    M_1.
    \label{eq:remaining_hadamards}
\end{equation*}

We now show that $H_R M_2 H_R$ can be converted to a single layer of gates in $\cG_{\rm elem}$ by computing the conjugation case by case. Since $M_2$ satisfies the hypothesis of Lemma \ref{lem:coordinate}, the only possibilities we must check are the preserved cases in Table \ref{tab:CNOT-CZ}. For example, a $\CNOT_{c\to t}$ may be conjugated by $H_t$ or $H_cH_t$, while conjugation by $H_c$ is forbidden. Repeating the analogous argument for $\CZ$, the relevant conjugation rules are:
\begin{align*}
    H_t\mathrm{CNOT}_{c\to t}H_t
        &=\mathrm{CZ}_{ct},\\
    H_cH_t\mathrm{CNOT}_{c\to t}H_cH_t
        &=\mathrm{CNOT}_{t\to c},\\
        H_j\mathrm{CZ}_{ij}H_j
        &=\mathrm{CNOT}_{i\to j}.
\end{align*}
The Pauli conjugation rules also remain in the same elementary gate set:
\begin{equation*}
    HXH=Z,\qquad HZH=X,\qquad HYH=-Y,
\end{equation*}
thus $H_SM_2H_S$ can be converted to a single layer of non-Hadamard elementary gates.
    
\end{proof}

\section{CNOT normal form for Hadamard-free depth-three circuits}
\label{app:cnot_normal_form}

In this section we prove that any optimal circuit implementing $U_G$ over $\cG_{\rm elem}$ can be converted to a circuit with no Hadamard and CNOT gates, without changing its depth.

By Lemma~\ref{lem:app-had-elim}, we may consider a Hadamard-free realization $U_G=L_3L_2L_1$. Since each layer contains no Hadamard gate, their action on computational-basis states is
\begin{equation*}
    L_t\lvert x\rangle
    =
    \mathrm{i}^{p_t(x)}
    \lvert A_tx+b_t\rangle,
    \label{eq:hadamard_free_basis_action}
\end{equation*}
where the binary invertible matrix $A_t$ is induced by the CNOT gates, the constant term $b_t$ is induced by $X$ and $Y$
gates, and $p_t(x)\in\{0,1,2,3\}$ records the phase induced by $Z,Y,S,S^\dagger$ and CZ gates.

Decompose the linear part as $A_t=I+N_t$, where $I$ denotes the identity matrix. Since the CNOTs of a layer act on disjoint pairs of qubits, each non-zero entry $(N_t)_{ij}$ corresponds to a $\CNOT_{j\rightarrow i}$ in layer $t$. Also, we have $A_t^2=I$ and $N_t^2=0$, since CNOT is its own inverse and inside each layer each pair of qubits is acted upon by at most one CNOT.

The target $U_G$ leaves every computational-basis string unchanged, which implies the CNOT layers must satisfy
\begin{equation*}
    A_3A_2A_1=I.
    \label{eq:three_layer_linear_identity}
\end{equation*}

We now restate and give a proof of the normal form in Lemma~\ref{lem:methods-cnot-normal}.

\begin{lemma}[CNOT normal form]
\label{lem:app-cnot-normal}
A Hadamard-free realization of $U_G$ of depth at most three can be simplified, without increasing depth, so that its middle layer contains no CNOTs and its outer layers contain identical CNOTs with the same orientation.
\end{lemma}

\begin{proof}
From the preceding discussion
\begin{align*}
    A_3
    =
    A_1^{-1}A_2^{-1}
    =
    (I+N_1)(I+N_2),
\end{align*}
and therefore
\begin{equation}
    N_3=N_1+N_2+N_1N_2.
    \label{eq:n3_matrix_relation}
\end{equation}

We claim that $N_1N_2=0$. In fact, suppose otherwise that $(N_1N_2)_{r,s}=1$. Since each CNOT act on disjoint pair of qubits inside each layer, there is a unique index $k$ such that
\begin{equation*}
    (N_1)_{r,k}=1
    \qquad \text{and} \qquad
    (N_2)_{k,s}=1.
    \label{eq:length_two_cnot_path}
\end{equation*}
These entries represent
$\mathrm{CNOT}_{k\to r}$ in layer $1$ and
$\mathrm{CNOT}_{s\to k}$ in layer $2$.
Since the second-layer CNOT already acts on $k$, the matching property
gives $(N_2)_{r,k}=0$. Moreover,
\begin{equation*}
    (N_1N_2)_{r,k}=0,
\end{equation*}
because the unique nonzero entry in row $r$ of $N_1$ is in column
$k$, while $(N_2)_{k,k}=0$. Equation~\eqref{eq:n3_matrix_relation}
therefore implies
\begin{equation*}
    (N_3)_{r,k}=1.
\end{equation*}
By an analogous argument, $(N_3)_{k,s}=1$, which is a contradiction, as it would imply two CNOT gates sharing a qubit in layer $3$. Therefore $N_1N_2=0$.

It follows that
\begin{equation*}
    N_3=N_1+N_2,
    \label{eq:n3_symmetric_difference}
\end{equation*} 
meaning that each CNOT in layer $2$ have a corresponding copy either on layer $1$ or in layer $3$. After canceling each CNOT in layer $2$ with a neighboring one, we have $N_2=0$ and $N_1=N_3$.
\end{proof}

We now assume the graph is $3$-regular to completely eliminate CNOT gates in the optimal circuit.

\begin{lemma}[CNOT elimination for 3-regular graphs]
\label{lem:app-3-regular}
If $G$ is 3-regular, the normal form of
Proposition~\ref{lem:app-cnot-normal} contains
no CNOTs.
\end{lemma}

\begin{proof}
Apply Proposition~\ref{lem:app-cnot-normal} to obtain a Hadamard-free realization
\begin{equation*}
    U_G=L_3L_2L_1
\end{equation*}
whose middle layer contains no CNOTs and whose outer
layers contain identical CNOTs with the same orientations.

We prove by contradicion: Suppose there is a $C:=\mathrm{CNOT}_{a\to b}$ for some pair of qubits in layer $3$. Write $L_3=CR_3$. Since gates within a layer have disjoint supports, $R_3$ acts trivially on both $a$ and $b$. Multiplying the circuit identity on the left by $C$ gives
\begin{equation}
    CU_G=R_3L_2L_1.
    \label{eq:3-regular-cancel}
\end{equation}
We compare how each unitary propagates the single-site Pauli operator $X_b$. Using the LHS, conjugation gives
\begin{align*}(CU_G)X_b(CU_G)^\dagger&=C\left(X_b\prod_{v\in N_G(b)}Z_v\right)C^\dagger\\
    &=X_b\prod_{v\in N_G(b)}Z_v
\end{align*}
where $N_G(b)$ is the neighborhood of $b$. The first equality follows from usual graph state preparation identities and the second follows from the fact that $C$ preserves $X_b$ and $Z_a$, and acts trivially on every other $Z_v$ appearing
in this product. Thus $X_b$ is propagated to another Pauli with weight $1+\deg_G(b)$. 

On the other hand, using the RHS of Eq.~\eqref{eq:3-regular-cancel}, the final propagated Pauli can have weight at most $3$. In fact, since the only gate in $L_1$ that acts on $X_b$ is $C$, we have:
\begin{align}(R_3L_2L_1)X_b(R_3L_2L_1)^\dagger&=R_3(L_2X_bL_2^\dagger)R_3^\dagger
\label{eq:3-regular-circuit-pauli}
\end{align}
Now, the middle layer contains neither Hadamards nor CNOTs, so its only possible actions on $X_b$ are
\begin{equation*}
    L_2X_bL_2^\dagger=
    \begin{cases}
        \pm X_b, & \text{Pauli or identity on }b,\\
        \pm Y_b, & S_b\text{ or }S_b^\dagger,\\
        X_bZ_j,  & \CZ_{bj}.
    \end{cases}
\end{equation*}

In the first two cases, conjugation by $R_3$ leaves the Pauli matrix unchanged, because $R_3$ has no gate acting on qubit $b$. Hence $X_b$ is propagated to a weight-$1$ Pauli operator. In the last case, Eq.~\eqref{eq:3-regular-circuit-pauli} becomes
\begin{equation*}
\bigl(R_3X_bZ_jR_3^\dagger\bigr)=X_b\bigl(R_3Z_jR_3^\dagger\bigr).
\end{equation*}
Since qubit $j$ can only be acted upon a single- or two-qubit gate from the last layer, $R_3Z_jR_3^\dagger$ is supported on at most two qubits. Thus, in every case, $X_b$ is propagated to a Pauli with weight $\leq 3$. Comparing with the previous computation, we must conclude
\begin{equation*}
    1+\deg_G(b)\le 3,
\end{equation*}
and therefore $\deg_G(b)\le 2$. This contradicts the initial assumption that $G$ is 3-regular.
\end{proof}

We now provide the final argument after Lemma~\ref{lem:cnot-elim} in the main text. After applying Lemmas~\ref{lem:app-had-elim}, \ref{lem:app-cnot-normal} and \ref{lem:app-3-regular}, only single-qubit Clifford gates and CZ gates remain. Since at each step the register is in the computational basis, the circuit action may be written as
\begin{equation*}
    \ket{x}\mapsto \mathrm{i}^{\ell(x)}(-1)^{q(x)}\ket{x+b},
\end{equation*}
 where $q(x)$ contains only quadratic terms corresponding exclusively to $CZ$ gates, $Y,Z,S,S^\dagger,\CZ$ gates may contribute to the linear terms $\ell(x)$, and $X,Y$ may contribute to the linear term $b$.

Let $m_{uv}$ denote the total number of occurrences of
$\mathrm{CZ}_{uv}$ in the three layers, and $\overline{m}_{uv}\in\{0,1\}$, its parity. The quadratic part of the circuit phase is therefore
\begin{equation*}
    q(x)
    =
    \sum_{u<v}
    \overline{m}_{uv}x_ux_v.
    \label{eq:circuit_quadratic_part}
\end{equation*}
Equality with $U_G$ implies
\begin{equation*}
    \overline{m}_{uv}
    =
    \begin{cases}
        1, & uv\in E,\\
        0, & uv\notin E.
    \end{cases}
    \label{eq:cz_parity_matches_graph}
\end{equation*}
In particular, the circuit contains at least one controlled-$Z$ gate for every edge of $G$. Since $G$ is 3-regular, and denoting by $N_{\mathrm{CZ}}$ the number of CZ gates,
\begin{equation*}
    N_{\mathrm{CZ}}
    \geq
    \sum_{u<v}m_{uv}
    \geq |E|
    =
    \frac{3n}{2}.
    \label{eq:cz_lower_bound}
\end{equation*}

Since each layer contains at most $n/2$ two-qubit gates, we must have exactly $N_{\mathrm{CZ}}=3n/2$ and no other gate appears. Consequently, $U_G$ is realized as a CZ-only circuit in depth $3$, which concludes $D_{\rm elem}(U_G)=D_{\rm CZ}(U_G)$.

\end{document}